\documentclass{sig-alternate-10pt-benny2}

\usepackage{amsmath,amssymb}
\usepackage{stmaryrd}
\usepackage{algorithmic}
\usepackage[usenames]{color}
\usepackage{graphics}
\usepackage{xspace}
\usepackage{url}
\usepackage{cite}
\usepackage{paralist}

\usepackage{enumitem}
\setlist[itemize]{itemsep=2pt, topsep=2pt}
\setlist[enumerate]{itemsep=2pt, topsep=2pt}

\newcommand{\eat}[1]{} \def\e#1{\emph{#1}}

\def\allreals{\mathbb{R}}
\def\allnaturals{\mathbb{N}}
 
\def\set#1{\mathord{\{#1\}}}

\def\eqdef{\mathrel{\stackrel{\textsf{\tiny def}}{=}}}
 \def\conf{\mathit{conf}}

\def\pardim{\mathsf{pardim}}

\newtheorem{theorem}{Theorem}[section]
\newtheorem{definition}[theorem]{Definition}

\newtheorem{proposition}[theorem]{Proposition}
\newtheorem{lemma}[theorem]{Lemma}
\newtheorem{examplethm}[theorem]{Example}
\newenvironment{example}{\begin{examplethm}\em}
{\qed\end{examplethm}}

\newenvironment{repeatresult}[2]
{\vskip0.5em\par\textsc{#1 #2.\!\!\!}\em}
{\vskip1em}

\newenvironment{repproposition}[1]{\begin{repeatresult}{Proposition}{#1}}{\end{repeatresult}}
\newenvironment{reptheorem}[1]{\begin{repeatresult}{Theorem}{#1}}{\end{repeatresult}}

\newcommand{\algname}[1]{{\sf #1}}

\def\myrulewidth{3.20in}

\def\therule{\rule{\myrulewidth}{0.2pt}}

\newenvironment{insidecode}[3]
{
\begin{tabular}{p{\myrulewidth}}
\multicolumn{1}{c}{\rule{0mm}{3mm}{\bf #3} $\algname{#1}(\mbox{#2})$\vspace{-0.6em}}\\
\therule\vskip-0.8em\therule
\vspace{-1em}
\begin{algorithmic}[1]}
{\end{algorithmic}
\vskip-0.3em\therule
\end{tabular}}

\def\D{\mathcal{D}}
\def\arity{\mathrm{arity}}
\def\tup#1{\mathbf{#1}}

\def\scs{\mathcal{S}}
\def\P{\mathcal{P}}
\def\G{\mathcal{G}}

\def\E{\mathcal{E}}
\def\I{\mathcal{I}}
\def\D{\mathcal{D}}
\def\F{\mathcal{F}}

\def\arity{\mathsf{arity}}
\def\rel#1{\mathrm{#1}}
\def\att#1{\mathit{#1}}

\def\lfp{\textsf{min-sol}}
\def\flfp{\textsf{min-sol}^{\mathsf{fin}}}
\def\inflfp{\textsf{min-sol}}

\def\distname#1{\mathsf{#1}}

\def\dla{\,\,\leftarrow\,\,}

\def\dcom{\,,\,}

\def\val#1{\text{$\mathsf{#1}$}}

\def\weight{\mathrm{weight}}

\def\leaves{\mathrm{leaves}}

\newcounter{tbsnr}
\newenvironment{tbs}
{\addtocounter{tbsnr}{1}\par\bigskip \noindent\fbox{\thetbsnr}
\hspace{2mm}\begin{minipage}{.9\linewidth}\tt \small}
{\end{minipage}\hspace*{\fill}\bigskip}

\newcommand{\tododone}[1]{}

\newcommand{\EDatalog}{\text{\textup{Datalog}$^\exists$}\xspace}
\newcommand{\DDatalog}{\text{GDatalog\textup{[}$\Delta$\textup{]}}\xspace}
\newcommand{\sol}{\Omega}%{\textup{outcomes}}
\newcommand{\fsol}{\sol^{\mathsf{fin}}}%{\textup{outcomes}^{\mathsf{fin}}}
\newcommand{\csol}[1]{\sol^{#1\subseteq}}%{\textup{outcomes}^{#1\subseteq}}
\newcommand{\maxpaths}{\mathit{paths}}
\def\prw{\mathbf{P}}

\def\partitle#1{\vskip1em \noindent \textbf{#1.}\,\,\,\,}

\newenvironment{citemize}
{\vskip0.4em\begin{compactitem}}
{\end{compactitem}\vskip0.4em}
\begin{document}

%\conferenceinfo{??}{??}

\title{Declarative Statistical Modeling with Datalog}
%\subtitle{\textsf{\large [\textit{New Formal Framework} Paper]}}

\def\ausp{\alignauthor}%\quad\quad}
\def\aff#1{\vskip0.1em\textrm{#1}}
\numberofauthors{5} 
\author{
\alignauthor 
Vince Barany
\aff{LogicBlox, Inc.}
%\and
\ausp
Balder ten Cate
\aff{LogicBlox, Inc.}
%\\ \& UC Santa Cruz
\ausp 
Benny Kimelfeld
\aff{LogicBlox, Inc.}
%\and
\ausp
Dan Olteanu
\aff{Oxford University\\
\& LogicBlox, Inc.
}
%\and
\ausp
Zografoula Vagena
\aff{LogicBlox, Inc.}
}

\maketitle

\begin{abstract}
\parbox{3.3in}{\small
Formalisms for specifying general statistical models, such as
probabilistic-programming languages, typically consist of two
components: a specification of a stochastic process (the prior), and a
specification of observations that restrict the probability space to a
conditional subspace (the posterior).  Use cases of such formalisms
include the development of algorithms in machine learning and
artificial intelligence. We propose and investigate a declarative
framework for specifying statistical models on top of a database,
through an appropriate extension of Datalog.  By virtue of extending
Datalog, our framework offers a natural integration with the database,
and has a robust declarative semantics (that is, semantic independence
from the algorithmic evaluation of rules, and semantic invariance
under logical program transformations).

Our proposed Datalog extension provides convenient mechanisms to
include common numerical probability functions; in particular,
conclusions of rules may contain values drawn from such functions.  The
semantics of a program is a probability distribution over the possible
outcomes of the input database with respect to the program; these
possible outcomes are minimal solutions with respect to a related
program that involves existentially quantified variables in
conclusions. Observations are naturally incorporated by means of
integrity constraints over the extensional and intensional
relations. We focus on programs that use discrete numerical
distributions, but even then the space of possible outcomes may be
uncountable (as a solution can be infinite).  We define a probability
measure over possible outcomes by applying the known concept of
cylinder sets to a probabilistic chase procedure. We show that the
resulting semantics is robust under different chases. We also identify
conditions guaranteeing that all possible outcomes are finite (and
then the probability space is discrete).  We argue that the framework
we propose retains the purely declarative nature of Datalog, and
allows for natural specifications of statistical models.
}
\end{abstract}

%%%%%%%%%%%%%%%%%%%%%%%%%%%%%%%%%%
%%%%%%%%%%%%%%%%%%%%%%%%%%%%%%%%%%
%\input{sections/introduction.tex}
\section{Introduction}
Formalisms for specifying general statistical models are commonly used
for developing machine learning and artificial intelligence algorithms
for problems that involve inference under uncertainty. A substantial
scientific effort has been made on developing such formalisms and
corresponding system implementations. An intensively studied concept
in that area is that of \e{Probabilistic
  Programming}~\cite{DBLP:conf/popl/Goodman13} (PP), where the idea is
that the programming language allows for building general random
procedures, while the system \e{executes} the program not in the
standard programming sense, but rather by means of \e{inference}.
Hence, a PP system is built around a language and an inference engine
(which is typically based on variants of \e{Markov Chain Monte Carlo},
most notably \e{Metropolis-Hastings}).  An inference task is a
probability-aware aggregate operation over all the possible worlds,
such as finding the most likely possible world, or estimating the
probability of an event (which is phrased over the outcome of the
program). Recently, DARPA initiated the project of \e{Probabilistic
  Programming for Advancing Machine Learning}, aimed at advancing PP
systems (with a focus on a specific collection of systems,
e.g.,~\cite{pfeffer2009figaro,DBLP:journals/corr/MansinghkaSP14,BLOG:IJCAI:2005})
towards facilitating the development of algorithms based on machine
learning.

In probabilistic programming, a statistical model is typically phrased
by means of two components. The first component is a \e{generative
  process} that produces a random possible world by straightforwardly
following instructions with randomness, and in particular, sampling
from common numerical probability functions; this gives the \e{prior
  distribution}. The second component allows to phrase constraints
that the relevant possible worlds should satisfy, and, semantically,
transforms the prior distribution into the \e{posterior
  distribution}---the subspace conditional on the constraints.

As an example, in supervised text classification (e.g., spam
detection) the goal is to classify a text document into one of several
known classes (e.g., spam/non-spam). Training data consists of a
collection of documents labeled with classes, and the goal of learning
is to build a model for predicting the classes of unseen documents.
One common approach to this task assumes a generative process that
produces random \e{parameters} for every class, and then uses these
parameters to define a generator of random words in documents of the
corresponding
class~\cite{DBLP:dblp_journals/ml/NigamMTM00,DBLP:bibsonomy_mccallum-multilabel}.
So, the prior distribution generates parameters and documents for each
class, and the posterior is defined by the actual documents of the
training data. In \e{unsupervised} text classification the goal is to
cluster a given set of documents, so that different clusters
correspond to different topics (which are not known in
advance). Latent Dirichlet
Allocation~\cite{DBLP:journals/jmlr/BleiNJ03} approaches this problem
in a similar generative way as the above, with the addition that each
document is associated with a distribution over topics.

While the agenda of probabilistic programming is the deployment of
programming languages to developing statistical models, in this
framework paper we explore this agenda from the point of view of
database programming.  Specifically, we propose and investigate an
extension of Datalog for declarative specification of statistical
models on top of a database. We believe that Datalog can be naturally
extended to a language for building statistical models, since its
essence is the production of new facts from known (database) facts. Of
course, traditionally these facts are deterministic, and our extension
enables the production of probabilistic facts that, in particular,
involve numbers from available numerical distributions.  And by virtue
of extending Datalog, our framework offers a natural integration with
the database, and has a robust declarative semantics: a program is a
set of rules that is semantically invariant under transformations that
retain logical equivalence. Moreover, the semantics of a program
(i.e., the probability space it specifies) is fully determined by the
satisfaction of rules, and does not depend on the specifics of any
execution engine.

In par with languages for probabilistic programming, our proposed
extension consists of two parts: a \e{generative Datalog program} that
specifies a prior probability space over (finite or infinite) sets of
facts that we call \e{possible outcomes}, and a definition of the
posterior probability by means of \e{observations}, which come in the
form of an ordinary logical constraint over the extensional and
intensional relations.

The generative component of our Datalog extension provides convenient
mechanisms to include conventional parameterized numerical probability
functions (e.g., Poisson, geometrical, etc.). Syntactically, this
extension allows to sample values in the conclusion of rules,
according to specified parameterized distributions. As an example,
consider the relation
$\mathrm{Client}(\mathrm{ssn},\mathrm{branch},\mathrm{avgVisits})$
that represents clients of a service provider, along with their
associated branch and average number of visits (say, per month). The
following distributional rule models a random number of visits for
that client in the branch.
\begin{equation}\label{eq:visits}
\mathrm{Visits}(c,b,\distname{Poisson}[\lambda])\leftarrow 
\mathrm{Client}(c,b,\lambda)
\end{equation}
Note, however, that a declarative interpretation of the above rule is
not straightforward. Suppose that we have another rule of the
following form:
\begin{equation}\label{eq:pref-visits}
\mathrm{Visits}(c,b,\distname{Poisson}[\lambda])\leftarrow 
\mathrm{PreferredClient}(c,b,\lambda)
\end{equation}
Then, what would be the semantics if a person is both a client and a
preferred client? Do we sample twice for that person? And what if the
two $\lambda$s of the two facts are not the same? Is sampling
according to one rule considered a satisfaction of the other rule?
What if we have also the following rule:
\begin{equation}\label{eq:visits-active}
\mathrm{Visits}(c,b,\distname{Poisson}[\lambda])\leftarrow 
\mathrm{Client}(c,b,\lambda),\mathrm{Active}(b)
\end{equation}
From the viewpoint of Datalog syntax, Rule~\eqref{eq:visits-active} is
logically implied by Rule~\eqref{eq:visits}, since the premise of
Rule~\eqref{eq:visits-active} implies the premise of
Rule~\eqref{eq:visits}. Hence, we would like the addition of
Rule~\eqref{eq:visits-active} to have no effect on the program. This
means that some rule instantiations will not necessarily fire an
actual sampling.

To make sense of rules such as the above, we associate with every
program $\G$ an auxiliary program $\widehat{\G}$, such that
$\widehat{\G}$ does not use distributions, but is rather an ordinary
Datalog program where a rule can have an existentially quantified
variable in the conclusion. Intuitively, in our example such a rule
states that ``if the premise holds, then there exists a fact
$\mathrm{Visits}(c,b,x)$ where $x$ is associated with the distribution
$\distname{Poisson}$ and the parameter $\lambda$.''  In particular, if
the program contains the aforementioned Rule~\eqref{eq:visits}, then
Rule~\eqref{eq:visits-active} has no effect; similarly, if the tuple
$(c,b,\lambda)$ is in both $\mathrm{Client}$ and
$\mathrm{PreferredClient}$, then in the presence of
Rule~\eqref{eq:pref-visits} the outcome does not change if one of
these tuples is removed.

In this paper we focus on numerical probability distributions that are
discrete (e.g., the aforementioned ones). Our framework has a natural
extension to continuous distributions (e.g., Gaussian or Pareto), but
our analysis requires a nontrivial generalization that we defer to
future work.

When applying the program $\G$ to an input instance $I$, the
probability space is over all the minimal solutions of $I$
w.r.t.~$\widehat{\G}$, such that all the numerical samples have a
positive probability.  To define the probabilities of a sample in this
probability space, we consider two cases. In the case where all the
possible outcomes are finite, we get a discrete probability
distribution, and the probability of a possible outcome can be defined
immediately from its content.  But in general, a possible outcome can
be infinite, and moreover, the set of all possible outcomes can be
uncountable. Hence, in the general case we define a probability
measure space. To make the case for the coherence of our definitions
(i.e., our definitions yield proper probability spaces), we define a
natural notion of a \e{probabilistic chase} where existential
variables are produced by invoking the corresponding numerical
distributions. We use \e{cylinder sets}~\cite{Ash2000} to define a
measure space based on a chase, and prove that this definition is
robust, since one establishes the same probability measure no matter
which chase is used.

%%%%%%%%%%%%%%%%%%%%%%%%%%%%%%
\partitle{Related Work} Our contribution is a marriage between
probabilistic programming and the declarative specification of
Datalog.  The key features of our approach are the ability to express
probabilistic models \e{concisely} and \e{declaratively} in a Datalog
extension with probability distributions as first-class citizens.
Existing formalisms that associate a probabilistic interpretation with
logic are either not declarative (at least in the Datalog sense) or
depart from the probabilistic programming paradigm (e.g., by lacking
the support for numerical probability distributions). We next discuss
representative related formalisms and contrast them with our work.
They can be classified into three broad categories: \e{(1)} imperative
specifications over logical structures, \e{(2)} logic over
probabilistic databases, \e{and (3)} indirect specifications over the
Herbrand base. (Some of these formalisms belong to more than one
category.)

The first category includes imperative probabilistic programming
languages~\cite{PP2014}, such as BLOG~\cite{BLOG:IJCAI:2005}, that can
express probability distributions over logical structures, via
generative stochastic models that can draw values at random from
numerical distributions, and condition values of program variables on
observations. In contrast with closed-universe languages such as SQL
and logic programs, BLOG considers open-universe probability models
that allow for uncertainty about the existence and identity of
objects. Instantiations of this category also do not focus on a
declarative specification, and indeed, their semantics is dependent on
their particular imperative implementations.  P-log~\cite{Baral:2009}
is a Prolog-based language for specifying Bayesian networks. Although
declarative in nature, the semantics inherently assumes a form of
acyclicity that allows the rules to be executed serially. Here we are
able to avoid such an assumption since our approach is based on the
minimal solutions of an existential Datalog program.

The formalisms in the second category view the generative part of the
specification of a statistical model as a two-step process. In the
first step, facts are being randomly generated by a mechanism external
to the program. In the second step, a logic program, such as
Prolog~\cite{ProbLog11} or Datalog~\cite{AbiteboulDV14}, is evaluated
over the resulting random structure. This approach has been taken by
PRISM~\cite{DBLP:conf/ijcai/SatoK97}, the \e{Independent
  Choice Logic}~\cite{DBLP:conf/ilp/Poole08}, and to a large extent by
\e{probabilistic databases}~\cite{PDB11} and their semistructured
counterparts~\cite{DBLP:series/sfsc/KimelfeldS13}.  The focus of our
work, in contrast, is on a formalism that completely defines the
statistical model, without referring to external processes.

One step beyond the second category and closer to our work is taken by
uncertainty-aware query languages for probabilistic data such as
TriQL~\cite{WidomTrio2008}, I-SQL, and world-set
algebra~\cite{WSA07,PWSA07}. The latter two are natural analogs to SQL
and relational algebra for the case of incomplete information and
probabilistic data~\cite{WSA07}. They feature constructs such as
\verb=repair-key=, \verb=choice-of=, \verb=possible=, \verb=certain=,
and \verb=group-worlds-by= that can construct possible worlds
representing all repairs of a relation with respect to (w.r.t.) key
constraints, close the possible worlds by unioning or intersecting
them, or group the worlds into sets with the same results to
sub-queries. World-set algebra has been extended to (world-set)
Datalog, fixpoint, and while-languages~\cite{DeutchKM10} to define
Markov chains. While such languages cannot explicitly specify
probability distributions, they may simulate a specific categorical
distribution indirectly using non-trivial programs with specialized
language constructs like \verb=repair-key= on input tuples with
weights representing samples from the distribution.

\eat{Key properties of world-set algebra are conservativity over
  relational algebra~\cite{WSA07}, genericity (in contrast to TriQL),
  expressiveness~\cite{Koch09}, and efficient
  processing~\cite{PWSA07}.  World-set algebra captures exactly
  second-order logic over finite structures, or equivalently, the
  polynomial hierarchy~\cite{Koch09}.  Moreover, it is closed under
  composition~\cite{Koch09}: view definitions do not add to the
  expressive power of the language even though first materializing a
  view and subsequently using it multiple times in the query is
  semantically different from composing the query with the view and
  thus obtaining several copies of the view definition that can
  independently make their non-deterministic choices. This problem may
  occur for rule inlining optimizations, which preserve equilavence of
  standard Datalog programs, in an extension of our formalism, where
  rules may have $\Delta$-atoms in the body. For instance, assume a
  rule with a $\Delta$-atom in the head is used twice in the body of
  another rule. Then, the set of possible worlds of this program may
  not the same with that of the program where we'd inline the first
  rule twice in the second rule.}

MCDB~\cite{MCDB2008} and SimSQL~\cite{SimSQL2013} propose SQL
extensions (with for-loops and probability distributions) coupled with
Monte Carlo simulations and parallel database techniques for
stochastic analytics in the database. In contrast, our work focuses on
existential Datalog with recursion and probability spaces over the
minimal solutions of the data w.r.t.~the Datalog program.

Formalisms in the third category are indirect specifications of
probability spaces over the \e{Herbrand base}, which is the set of all
the facts that can be obtained using the predicate symbols and the
constants of the database. This category includes \e{Markov Logic
  Networks
  (MLNs)}~\cite{DBLP:series/synthesis/2009Domingos,DBLP:journals/pvldb/NiuRDS11},
where the logical rules are used as a compact and intuitive way of
defining \e{factors}.  In other words, the probability of a possible
world is the product of all the numbers (factors) that are associated
with the rules that the world satisfies.  This approach is applied in
DeepDive~\cite{DBLP:conf/vlds/NiuZRS12}, where a database is used for
storing relational data and extracted text, and database queries are
used for defining the factors of a factor graph.  We view this
approach as \e{indirect} since a rule does not determine directly the
distribution of values. Moreover, the semantics of rules is such that
the addition of a rule that is logically equivalent to (or implied by,
or indeed equal to) an existing rule changes the semantics and thus
the probability distribution. A similar approach is taken by
\e{Probabilistic Soft Logic}~\cite{DBLP:conf/uai/BrochelerMG10}, where
in each possible world every fact is associated with a weight (degree
of truth).

Further formalisms in this category are \e{probabilistic
  Datalog}~\cite{DBLP:journals/jasis/Fuhr00}, \e{probabilistic
  Datalog+/-}~\cite{PDatalogPlusMinus13}, and {\em probabilistic logic
  programming (ProbLog)}~\cite{ProbLog11}. In these formalisms, every
rule is associated with a probability. For ProbLog, the semantics is
not declarative as the rules follow a certain evaluation order; for
probabilistic Datalog, the semantics is purely declarative. Both
semantics are different from ours and that of the other formalisms
mentioned thus far. A Datalog rule is interpreted as a rule over a
probability distribution over possible worlds, and it states that, for
a given grounding of the rule, the marginal probability of being true
is as stated in the rule. Probabilistic Datalog+/- uses MLNs as the
underlying semantics.  Besides our support for numerical probability
distributions, our formalism is used for defining a single probability
space, which is in par with the standard practice in probabilistic
programming. 

As said earlier, the programs in our proposed formalism allow for
recursion. As we show in the paper, the semantics is captured by
Markov chains that may be infinite. Related formalisms are those of
the \e{Probabilistic Context-Free Grammar} (PCFG) and the more general
\e{Recursive Markov Chain}
(RMC)~\cite{DBLP:journals/jacm/EtessamiY09}, where the probabilistic
specification is by means of a finite set of transition graphs that
can call one another (in the sense of method call) in a possibly
recursive fashion. In database research, PCFGs and RMCs have been
explored in the context of probabilistic
XML~\cite{DBLP:conf/icdt/CohenK10,DBLP:journals/pvldb/BenediktKOS10}.
Although these formalisms do not involve numerical distributions, in
future work we plan to conduct a study of the relative expressive
power between them and restrictions of our framework. Moreover, we
plan to study whether and how inference techniques on PCFGs and RMCs
can be adapted to our framework.

\partitle{Organization}
The remainder of the paper is organized as follows. In
Section~\ref{sec:preliminaries} we give basic definitions.  The syntax
and semantics of \e{generative Datalog} is introduced in
Section~\ref{sec:generative}, where we focus on the case where all
solutions are finite.  In Section~\ref{sec:chase} we present our
adaptation of the chase. The general case of generative Datalog, where
solutions can be infinite, is presented in Section~\ref{sec:measures}.
We complete our development in Section~\ref{sec:ppdl}, where
generative Datalog is extended with constraints (observations) to form
\e{Probabilistic-Programming Datalog} (\e{PPDL}). Finally, we discuss
extensions and future directions in Section~\ref{sec:extensions} and
conclude in Section~\ref{sec:conclusions}.

%%%%%%%%%%%%%%%%%%%%%%%%%%%%%%%%%%
%%%%%%%%%%%%%%%%%%%%%%%%%%%%%%%%%%
%\input{sections/preliminaries.tex}
\section{Preliminaries}\label{sec:preliminaries}

In this section we give some preliminary definitions that we will use
throughout the paper.

%\subsection{Schemas and Instances}
\partitle{Schemas and instances}
A (\e{relational}) \e{schema} is a collection $\scs$ of \e{relation
  symbols}, where each relation symbol $R$ is associated with an
\e{arity}, denoted $\arity(R)$, which is a natural number.  An
\e{attribute} of a relation symbol $R$ is any number in
$\set{1,\dots,\arity(R)}$.  For simplicity, we consider here only
databases over real numbers; our examples may involve strings, which
we assume are translatable into real numbers. A
\e{fact} over a schema $\scs$ is an expression of the form $R(c_1,
\ldots, c_n)$ where $R$ is an $n$-ary relation in $\scs$ and $c_1,
\ldots, c_n\in \allreals$. An \e{instance} $I$ over $\scs$ is a finite
set of facts over $\scs$. We will denote by $R^I$ the set of all
tuples $(c_1, \ldots, c_n)$ such that $R(c_1, \ldots, c_n)\in I$ is a
fact of $I$.

%\subsection{Datalog Programs}
\partitle{Datalog programs}
In this work we use Datalog with the option of having existential
variables in the head~\cite{DatalogPlusMinus10}. Formally, an
\e{existential Datalog program}, or just \e{\EDatalog program} for
short, is a triple $\D=(\E,\I,\Theta)$ where:
%\begin{itemize}
%\begin{compactitem}
%\item 
(1) $\E$ is a schema, called the \e{extensional database} (EDB)
schema,
%\item 
(2) $\I$ is a schema, called the \e{intensional database} (IDB)
schema, and is disjoint from $\E$, 
%\item 
and (3)
$\Theta$ is a finite set of \emph{\EDatalog rules}, i.e.,,
first-order formulas of the form
\[\forall\tup x \big[~\exists\tup
y(\psi(\tup x,\tup y)) \leftarrow \varphi(\tup x)~\big]\] where
$\varphi(\tup x)$ is a conjunction of atomic formulas over $\E\cup\I$ and
$\psi(\tup x,\tup y)$ is an atomic formula over $\I$, such that
each variable in $\tup x$ occurs in at least one atomic formula of $\varphi$. 
%\end{itemize}
%\end{compactitem}
Here, by an \emph{atomic formula} (or, \emph{atom}) we mean an
expression of the form $R(t_1, \ldots, t_n)$ where $R$ is an $n$-ary
relation and $t_1, \ldots, t_n$ are either constants (i.e., real
numbers) or variables.  We usually omit the universal quantifiers for
readability's sake.  \e{Datalog} is the fragment of \EDatalog where the
conclusion (left-hand side) of each rule is a single atomic formula
without existential quantifiers.

Let $\D=(\E,\I,\Theta)$ be a \EDatalog program. An \e{input instance}
for $\D$ is an instance $I$ over $\E$.  A \e{solution} of $I$
w.r.t.~$\D$ is a possibly-infinite set $F$ of facts over $\E\cup\I$,
such that $I\subseteq F$ and $F$ satisfies all rules in $\Theta$
(viewed as first-order sentences).  A \e{minimal solution} of $I$
(w.r.t.~$\D$) is a solution $F$ of $I$ such that no proper subset of
$F$ is a solution of $I$.  The set of all, finite and infinite,
minimal solutions of $I$ w.r.t.~$\D$ is denoted by $\inflfp_{\D}(I)$,
and the set of all \e{finite} minimal solutions is denoted by
$\flfp_{\D}(I)$.  It is a well known fact that, if $\D$ is a
\emph{Datalog program} (that is, without existential quantifiers),
then every input instance $I$ has a unique minimal solution, which is
finite, and therefore $\flfp_{\D}(I)=\inflfp_{\D}(I)$.

%\subsection{Probability Spaces}
\partitle{Probability spaces}
We separately consider \e{discrete} and \e{continuous} probability
spaces.  We initially focus on the discrete case; there, a
\e{probability space} is a pair $(\Omega,\pi)$, where $\Omega$ is a
finite or countably infinite set, called the \e{sample space}, and
$\pi:\Omega\rightarrow[0,1]$ is such that
$\sum_{o\in\Omega}\pi(o)=1$. If $(\Omega,\pi)$ is a probability space,
then $\pi$ is a \e{probability distribution} over $\Omega$. We say
that $\pi$ is a \e{numerical} probability distribution if
$\Omega\subseteq\allreals$. In this work we focus on discrete
numerical distributions.

A \e{parameterized} probability distribution is a function
$\delta:\Omega\times\allreals^k\rightarrow[0,1]$, such that
$\delta(\cdot,\tup p):\Omega\rightarrow[0,1]$ is a probability
distribution for all $\tup p\in\allreals^k$.  We use $\pardim(\delta)$
to denote the number $k$, called the \e{parameter dimensionality} of
$\delta$. For presentation's sake, we may write $\delta(o|\tup p)$
instead of $\delta(o,\tup p)$. Moreover, we denote the
(non-parameterized) distribution $\delta(\cdot|\tup p)$ by
$\delta[\tup p]$.  Examples of (discrete) parameterized distributions
follow.

\begin{citemize}
\item $\distname{Flip}(x|p)$: $\Omega$ is $\set{0,1}$, and for a
  parameter $p\in[0,1]$ we have $\distname{Flip}(1|p)=p$ and
  $\distname{Flip}(0|p)=1-p$.
\item $\distname{Poisson}(x|\lambda)$: $\Omega=\allnaturals$, and for a
  parameter $\lambda\in(0,\infty)$ we have
  $\distname{Poisson}(x|\lambda)= \lambda^xe^{-\lambda}/x!$.
\item $\distname{Geo}(x|p)$: $\Omega=\allnaturals$, and for a
  parameter $p\in[0,1]$ we have $\distname{Geo}(x|p)=(1-p)^xp$.
%\item $\distname{Binom}(x|n,p)$: $\Omega=\set{0,1,\dots,n}$, and for
%  parameters $p\in[0,1]$ and $n\in\allnnreals$ we have
%  $\distname{Binom}(x|n,p)= \binom{n}{x} p^x(1-p)^{n-x}$.
\end{citemize}

In Section~\ref{sec:extensions} we will discuss the extension of our
framework to models that have an unbounded number of parameters, and
to continuous distributions.

%%%%%%%%%%%%%%%%%%%%%%%%%%%%%%%%%%
%%%%%%%%%%%%%%%%%%%%%%%%%%%%%%%%%%
%\input{sections/generative.tex}
\section{Generative Datalog}\label{sec:generative}
A Datalog program without existential quantifiers specifies how to
obtain a \e{solution} from an input EDB instance by producing the set
of inferred IDB facts.  In this section we present \e{generative
  Datalog programs}, which specify how to infer a \e{distribution}
over possible outcomes given an input EDB instance.

\subsection{Syntax}
We first define the syntax of a generative Datalog program, which we
call a \e{\DDatalog program}.
\begin{definition}[\DDatalog]
Let $\Delta$ be a finite set of parametrized numerical distributions. 
\begin{enumerate}
\item A \e{$\Delta$-term} is a term of the form $\delta[p_1, \ldots,
  p_k]$ where $\delta\in\Delta$ is a parametrized distribution with
  $\pardim(\delta)=k$, and $p_1, \ldots, p_k$ are variables and/or
  constants.
\item A \e{$\Delta$-atom} in a schema $\scs$ is an atomic formula
  $R(t_1, \ldots, t_n)$ with $R\in\scs$ an $n$-ary relation, such that
  exactly one term $t_i$ ($1\leq i\leq n$) is a $\Delta$-term, and all
  other $t_j$ are constants and/or variables.
\item A \e{\DDatalog rule} over a pair of disjoint schemas $\E$ and
  $\I$ is a first-order sentence of the form $\forall{\tup
    x}(\psi({\tup x})\leftarrow \phi({\tup x}))$ where $\phi({\tup
    x})$ is a conjunction of atoms in $\E\cup\I$ and $\psi({\tup x})$
  is either an atom in $\I$ or a $\Delta$-atom in $\I$.
\item A \e{\DDatalog program} is a triple $\G=(\E,\I,\Theta)$, where
  $\E$ and $\I$ are disjoint schemas and $\Theta$ is a finite set of
  \DDatalog rules over $\E$ and $\I$.
\end{enumerate}
\end{definition}

\begin{example}\label{example:generative}
  Our example is based on the burglar example of
  Pearl~\cite{DBLP:books/daglib/0066829} that has been frequently used
  for illustrating probabilistic programming
  (e.g.,~\cite{DBLP:conf/aaai/NoriHRS14}).  Consider the EDB schema
  $\E$ consisting of the following relations: $\rel{House}(h,c)$
  represents houses $h$ and their location cities $c$,
  $\rel{Business}(b,c)$ represents businesses $b$ and their location
  cities $c$, $\rel{City}(c,r)$ represents cities $c$ and their
  associated burglary rates $r$, and $\rel{AlarmOn}(x)$ represents
  units (houses or businesses) $x$ where the alarm is on.
  Figure~\ref{fig:I} shows an instance $I$ over this schema.  Now
  consider the \DDatalog program $\G=(\E,\I,\Theta)$ of
  Figure~\ref{fig:ddatalog}.
\begin{figure}
\fbox{\parbox{3.3in}{
\begin{enumerate}
  \item $\rel{Earthquake}(c,\distname{Flip}[0.01]) \dla \rel{City}(c,r)$
  \item $\rel{Unit}(h,c) \dla \rel{Home}(h,c)$
  \item $\rel{Unit}(b,c) \dla \rel{Business}(b,c)$
  \item $\rel{Burglary}(x,c,\distname{Flip}[r]) \dla \rel{Unit}(x,c)\dcom \rel{City}(c,r)$  
  \item $\rel{Trig}(x,\distname{Flip}[0.6]) \dla \rel{Unit}(x,c)\dcom\rel{Earthquake}(c,1)$
  \item $\rel{Trig}(x,\distname{Flip}[0.9]) \dla \rel{Burglary}(x,c,1)$
  \item $\rel{Alarm}(x) \dla \rel{Trig}(x,1)$
\end{enumerate}}}
\caption{\label{fig:ddatalog}An example \DDatalog program $\G$}
\end{figure}
Here, $\Delta$ consists of only one distribution, namely
$\distname{Flip}$.  The first rule above, intuitively, states that,
for every fact of the form $\rel{City}(c,r)$, there must be a fact
$\rel{Earthquake}(c,y)$ where $y$ is drawn from the Flip (Bernoulli)
distribution with the parameter $0.01$.
\end{example}

\subsection{Possible Outcomes}
To define the \e{possible outcomes} of a \DDatalog program, we associate to
each \DDatalog program $\G=(\E,\I,\Theta)$ a corresponding \EDatalog
program $\widehat{\G}=(\E,\I^\Delta,\Theta^\Delta)$.  The
possible outcomes of an input instance $I$ w.r.t.~$\G$ will then be minimal
solutions of $I$ w.r.t.~$\widehat{\G}$. Next, we describe 
$\I^\Delta$ and $\Theta^\Delta$.

The schema $\I^\Delta$ extends $\I$ with the following additional
relation symbols: whenever a rule in $\Theta$ contains a $\Delta$-atom
of the form $R(\ldots, \delta[\ldots], \ldots)$, and $i\leq\arity(R)$
is the argument position at which the $\delta$-term in question
occurs, then we add to $\I^\Delta$ a corresponding relation symbol
$R^\delta_i$, whose arity is $\arity(R)+\pardim(\delta)$.  These
relation symbols $R^\delta_i$ are called the \e{distributional}
relation symbols of $\I^\Delta$, and the other relation symbols of
$\I^\Delta$ (namely, those of $\I$) are referred to as the
\e{ordinary} relation symbols.  Intuitively, a fact in $R^\delta_i$
asserts the existence of a tuple in $R$ and a sequence of parameters,
such that the $i$th element of the tuple is sampled from $\delta$
using the parameters.

The set $\Theta^\Delta$ contains three kinds of rules:
\begin{enumerate}
\item[(i)] All Datalog rules from $\Theta$ that contain no
  $\Delta$-terms;
\item[(ii)] The rule $\exists y R^\delta_i({\tup t}, y, {\tup t'},
  {\tup p})\leftarrow \phi({\tup x})$ for every rule of the form
  $R({\tup t}, \delta[\tup p], {\tup t'})\leftarrow \phi({\tup x})$
  in $\Theta$, where $i$ is the position of $\delta[\tup p]$ 
  in $R({\tup t}, \delta[\tup p], {\tup t'})$;  
\item[(iii)] The rule $\forall{\tup x},{\tup p}
  (R(\tup{x})\leftarrow R^\delta_i({\tup x},{\tup p}))$ for every
  distributional relation symbol $R^\delta_i\in \I^\Delta$.
\end{enumerate}
Note that in (ii), $\tup t$ and $\tup t'$ are the terms that occur
before and after the $\Delta$-term $\delta[\tup p]$, respectively. A
rule in (iii) states that every fact in $R^\delta_i$ should be
reflected in the relation $R$.

\begin{example} 
  The \DDatalog program $\G$ given in Example~\ref{example:generative}
  gives rise to the corresponding $\EDatalog$ program $\widehat{\G}$
  of Figure~\ref{fig:edatalog}.  As an example of (ii), rule~6 of
  Figure~\ref{fig:ddatalog} is replaced with rule~6 of
  Figure~\ref{fig:edatalog}. Rules~8--10 of Figure~\ref{fig:edatalog}
  are examples of (iii).  
  % Note that $\widehat{\G}$ contains several
  % more rules of type (iii) for relation symbols $R^\delta_i$ where
  % $\delta$ never occurs in the $i$th position of $R$; those do not
  % affect the program, and are omitted from the example.
\end{example}

\begin{figure}
\scalebox{0.93}{\fbox{\parbox{3.5in}{
  \begin{enumerate}
  \item $\exists y~\rel{Earthquake}^\distname{Flip}_2(c,y,0.01) \dla \rel{City}(c,r)$
  \item $\rel{Unit}(h,c) \dla \rel{Home}(h,c)$
  \item $\rel{Unit}(b,c) \dla \rel{Business}(b,c)$
  \item $\exists y~\rel{Burglary}^\distname{Flip}_3(x,c,y,r) \dla \rel{Unit}(x,c)\dcom \rel{City}(c,r)$  
  \item \mbox{$\exists y~\rel{Trig}^\distname{Flip}_2(x,y,0.6) \dla\rel{Unit}(x,c)\dcom\rel{Earthquake}(c,1)$}
  \item $\exists y~\rel{Trig}^\distname{Flip}_2(x,y,0.9) \dla \rel{Burglary}(x,c,1)$
  \item $\rel{Alarm}(x) \dla \rel{Trig}(x,1)$
  \item $\rel{Earthquake}(c,d) \leftarrow \rel{Earthquake}^\distname{Flip}_2(c,d,p)$
  \item $\rel{Burglary}(x,y) \leftarrow \rel{Burglary}^\distname{Flip}_2(x,y,p)$
  \item $\rel{Trig}(x,y) \leftarrow \rel{Trig}^\distname{Flip}_2(x,y,p)$
\end{enumerate}}}}
\caption{\label{fig:edatalog}The \EDatalog program $\widehat{\G}$ for
  the \DDatalog program $\G$ of Figure~\ref{fig:ddatalog}}
\end{figure}
 
A \e{possible outcome} is defined as follows.
\begin{definition}[Possible Outcome]
  Let $I$ be an input instance for a \DDatalog program $\G$.  A
  \e{possible outcome} for $I$ w.r.t.~$\G$ is a  minimal solution $F$
  of $I$ w.r.t.~$\widehat{\G}$, such that $\delta(b|\tup p)> 0$ for
  every distributional fact $R^\delta_i({\tup a}, b, {\tup c}, {\tup
    p})\in F$ with $b$ in the $i$th position.
\end{definition}
We denote the set of all possible outcomes of $I$ w.r.t.~$\G$ by
$\sol_\G(I)$, and we denote the set of all finite possible outcomes by
$\fsol_\G(I)$.

The following proposition provides an insight into the possible
outcomes of an instance, and will reappear later on in our study of
the chase. For any distributional relation $R^\delta_i\in
\Theta^\Delta$, the \emph{functional dependency associated to
  $R^\delta_i$} is the functional dependency $R^\delta_i: (\{1,
\ldots, \arity(R^\delta_i)\}\setminus \{i\}) \to i$, expressing that
the $i$-th attribute is functionally determined by the rest.

\begin{proposition}\label{prop:sol:fd} 
  Let $I$ be any input instance for a \DDatalog instance $\G$. Then
  every possible outcome in $\sol_\G(I)$ satisfies all functional dependencies
  associated to distributional relations.
\end{proposition} 

The proof of Proposition~\ref{prop:sol:fd} is easy: if an instance $J$
violates the funtional dependency associated to a distributional
relation $R^\delta_i$, then one of the two facts involved in the
violation can be removed, showing that $J$ is, in fact, not a minimal solution
 w.r.t.~$\widehat{\G}$.

{
\begin{figure}[t]
\centering
\begin{tabular}[t]{cc}
  \multicolumn{2}{c}{$\rel{House}$}\\\hline
  $\att{id}$ & $\att{city}$ \\\hline
  \val{NP1} & \val{Napa}\\
  \val{NP2} & \val{Napa} \\
  \val{YC1} & \val{Yucaipa} \\
  \hline
\end{tabular}
\quad\,\,
\begin{tabular}[t]{cc}
  \multicolumn{2}{c}{$\rel{Business}$}\\\hline
  $\att{id}$ & $\att{city}$ \\\hline
  \val{NP3} & \val{Napa}\\
  \val{YC1} & \val{Yucaipa} \\
  \hline
\end{tabular}
\vskip1em
\begin{tabular}[t]{cc}
  \multicolumn{2}{c}{$\rel{City}$}\\\hline
  $\att{name}$ & $\att{burglary rate}$ \\\hline
  \val{Napa} & \val{0.03}\\
  \val{Yucaipa} & \val{0.01} \\
  \hline
\end{tabular}
\quad\,\,
\begin{tabular}[t]{c}
  \multicolumn{1}{c}{$\rel{AlarmOn}$}\\\hline
  $\att{unit}$ \\\hline
  \val{NP1} \\
  \val{YC1} \\
  \val{YC2} \\
  \hline
\end{tabular}
\caption{\label{fig:I}EDB instance of the running example}
\end{figure}
}

\subsection{Finiteness and Weak Acyclicity}
Our presentation first focuses on the case where all the possible
outcomes for the \DDatalog program are finite.  Before we proceed to
defining the semantics of such a \DDatalog program, we present the
notion of \e{weak acyclicity} for a \DDatalog program, as a natural
syntactic property that guarantees finiteness of all possible
outcomes. This draws on the notion of weak acyclicity for
\EDatalog~\cite{DBLP:conf/icdt/FaginKMP03}.  Consider any \DDatalog
program $\G=(\E,\I,\Theta)$.  A \e{position} of $\I$ is a pair $(R,i)$
where $R\in\I$ and $i$ is an attribute of $R$. The \e{dependency
  graph} of $\G$ is the directed graph that has the attributes of $\I$
as the nodes, and the following edges:
\begin{itemize}
\item A \e{normal edge} $(R,i)\rightarrow (S,j)$ whenever there is a
  rule $\psi(\tup x)\leftarrow\varphi(\tup x)$ and a variable $x$ at
  position $(R,i)$ in $\varphi(\tup x)$, and at position $(S,j)$ in
  $\psi(\tup x)$.
\item A \e{special edge} $(R,i)\rightarrow^* (S,j)$ whenever there is
  a rule of the form $$S(t_1, \ldots, t_{j-1},\delta[\tup p], t_{j+1},
  \ldots, t_n)\leftarrow\varphi(\tup x)$$ and an exported variable at
  position $(R,i)$ in $\varphi(\tup x)$. By an \e{exported variable},
  we mean a variable that appears in both the premise and the
  conclusion.
 \end{itemize}

We say that $\G$ is \e{weakly acyclic} if no cycle in the
dependency graph of $\G$ contains a special edge.

\def\thmweaklyacyclic{If a $\DDatalog$ program $\G$ is weakly acyclic,
  then $\sol_\G(I) = \fsol_\G(I)$ for all input instances $I$.  }

\begin{theorem}\label{thm:weakly-acyclic}
\thmweaklyacyclic
\end{theorem}

{
\begin{figure}[t]
\centering
\begin{tabular}[t]{cc|c||c}
  \multicolumn{4}{c}{$\rel{Earthquake}^{\distname{Flip}}_2$}\\\hline
  $\att{city}$ & $\att{draw}$ & param & $w(f)$\\\hline
  \val{Napa}   & 1             & $0.01$ & $0.01$ \\
  \val{Yucaipa}& 0             & $0.01$ & $0.99$ \\
  \hline
\end{tabular}
\quad
\begin{tabular}[t]{cc}
  \multicolumn{2}{c}{$\rel{Earthquake}$}\\\hline
  $\att{city}$ & $\att{draw}$ \\\hline
  \val{Napa}   & 1             \\ 
  \val{Yucaipa}& 0             \\
  \hline
\end{tabular}
\vskip1em
\begin{tabular}[t]{ccc|c||c}
  \multicolumn{5}{c}{$\rel{Burglary}^{\distname{Flip}}_3$}\\\hline
  $\att{unit}$ & $\att{city}$ & $\att{draw}$ & param  & $w(f)$\\\hline
  \val{NP1}     & \val{Napa}   & 1             & $0.03$ & $0.03$ \\
  \val{NP2}     & \val{Napa}   & 0             & $0.03$ & $0.97$ \\
  \val{NP3}     & \val{Napa}   & 1             & $0.03$ & $0.03$ \\
  \val{YU1}     &\val{Yucaipa} & 0             & $0.01$ & $0.99$ \\
  \hline
\end{tabular}
\quad
\begin{tabular}[t]{c}
  \multicolumn{1}{c}{$\rel{Alarm}$}\\\hline
  $\att{unit}$ \\\hline
  \val{NP1}     \\  
  \val{NP2}     \\  
  \hline
\end{tabular}
\vskip1em
\begin{tabular}[t]{ccc}
  \multicolumn{3}{c}{$\rel{Burglary}$}\\\hline
  $\att{unit}$ & $\att{city}$ & $\att{draw}$ \\\hline
  \val{NP1}     & \val{Napa}   & 1            \\
  \val{NP2}     & \val{Napa}   & 0            \\
  \val{NP3}     & \val{Napa}   & 1            \\
  \val{YU1}     &\val{Yucaipa} & 0            \\
  \hline
\end{tabular}
\quad
\begin{tabular}[t]{cc}
  \multicolumn{2}{c}{$\rel{Unit}$}\\\hline
  $\att{id}$ & $\att{city}$ \\\hline
  \val{NP1}   & \val{Napa}  \\  
  \val{NP2}   & \val{Napa}  \\  
  \val{NP3}   & \val{Napa}  \\      
  \val{YU1}   & \val{Yucaipa}\\
  \hline
\end{tabular}
\vskip1em
\begin{tabular}[t]{cc|c||c}
  \multicolumn{4}{c}{$\rel{Trig}^{\distname{Flip}}_2$}\\\hline
  $\att{unit}$ & $\att{draw}$ & param  & $w(f)$\\\hline
  \val{NP1}     & 1             & $0.9$  & $0.9$ \\
  \val{NP3}     & 0             & $0.9$  & $0.1$ \\
  \val{NP1}     & 1             & $0.6$  & $0.6$ \\
  \val{NP2}     & 1             & $0.6$  & $0.6$ \\
  \val{NP3}     & 0             & $0.6$  & $0.4$ \\
  \hline
\end{tabular}
\quad
\begin{tabular}[t]{cc}
  \multicolumn{2}{c}{$\rel{Trig}$}\\\hline
  $\att{unit}$ & $\att{draw}$ \\\hline
  \val{NP1}     & 1            \\
  \val{NP3}     & 0            \\
  \val{NP1}     & 1            \\
  \val{NP2}     & 1           \\
  \val{NP3}     & 0           \\
  \hline
\end{tabular}

\caption{\label{fig:idb}A possible outcome for the input instance $I$ in the running
  example}
\end{figure}
}

\subsection{Probabilistic Semantics}

Intuitively, the semantics of a $\DDatalog$ program is a function that
maps every input instance $I$ to a probability distribution over
$\sol_\G(I)$. We now make this precise.  Let $\G$ be a \DDatalog
program, let $I$ be an input for $\G$.  Again, we first consider the
case where an input instance $I$ only has finite possible outcomes (i.e.,
$\sol_{\G}(I)=\fsol_{\G}(I)$).  Observe that, when all possible outcomes of
$I$ are finite, the set $\sol_\G(I)$ is countable, since we assume
that all of our numerical distributions are discrete.  In this case,
we can define a discrete probability distribution over the possible outcomes
of $I$ w.r.t.~$\G$. We denote this probability distribution
by $\Pr_{\G,I}$.

For a distributional fact $f=R^\delta_i(a_1, \ldots, a_n, \tup p)$,
we define the \e{weight} of $f$ (notation: $\weight(f)$) to be $\delta(a_i|\tup p)$. 
For an ordinary (non-distributional) fact $f$, we set $\weight(f)=1$.
 For a finite set $F$ of facts, we denote by
$\prw(F)$ the product of the weights of all the facts in $F$.
\[\prw(F)\eqdef\prod_{f\in F}\weight(f)\]
The probability assigned to a possible outcome $J\in\fsol_P(I)$,
denoted $\Pr_{\G,I}(J)$, is simply $\prw(J)$. If a possible outcome
$J$ does not contain any distributional facts, then $\Pr_{\G,I}(J)=1$
by definition.

\begin{example}(continued) Let $J$ be the instance that consists of
  all of the relations in Figures~\ref{fig:I} and~\ref{fig:idb}.  Then
  $J$ is a possible outcome of $I$ w.r.t.~$\G$. For convenience, in the case
  of distributional relation symbols, we have added the weight of each
  fact to the corresponding row as the rightmost attribute. This
  weight is not part of our model (since it can be inferred from the
  rest of the attributes). For presentation's sake, the sampled values
  are under the attribute name $\att{draw}$ (while attribute names are
  again external to our formal model). $\Pr_{\G,I}(J)$ is the product
  of all of the numbers in the columns titled ``$w(f)$,'' that is,
  $0.01\times 0.99\times 0.03\times\cdots\times 0.4$.
\end{example}

The following theorem states that $\Pr_{\G,I}$ is indeed a probability
space over all the possible outcomes.

\begin{theorem}\label{thm:discrete}
  Let $\G$ be a \DDatalog program, and $I$ an input instance for $\G$,
  such that $\sol_{\G}(I)=\fsol_{\G}(I)$. Then $\Pr_{\G,I}$ is a
  discrete probability function over $\sol_\G(I)$.
\end{theorem}

We prove Theorem~\ref{thm:discrete} in Section~\ref{sec:chase}. In
Section~\ref{sec:measures} we consider the general case, and in
particular the generalization of Theorem~\ref{thm:discrete}, where not
all possible outcomes are guaranteed to be finite. There, if one
considers only the (countable set of all) finite possible outcomes,
then the sum of probabilities is not necessarily one. But still:

\begin{theorem}\label{thm:at-most-one}
  Let $\G$ be a \DDatalog program, and $I$ an input for $\G$. Then
  $\Sigma_{J\in\fsol_\G(I)}\Pr_{\G,I}(J)\leq 1$.
\end{theorem}

We conclude this section with some comments. First, we note that the
restriction of a conclusion of a rule to include a single
$\Delta$-term significantly simplifies the presentation, but does not
reduce the expressive power. In particular, we could simulate multiple
$\Delta$-terms in the conclusion using a collection of predicates and
rules. For example, if one wishes to have conclusion where a person
gets both a random height and a random weight (possibly with shared
parameters), then she can do so by deriving $\rel{PersonHeight}(p,h)$
and $\rel{PersonWeight}(p,w)$ separately, and using the rule
$\rel{PersonHW}(p,h,w)\leftarrow \rel{PersonHeight}(p,h),
\rel{PersonWeight}(p,w)$. We also highlight the fact that our
framework can easily simulate the probabilistic database model of
\e{independent tuples}~\cite{PDB11} with probabilities mentioned in
the database, using the $\distname{Flip}$ distribution, as follows.
Suppose that we have the EDB relation $R(x,p)$ where $p$ represents
the probability of every tuple. Then we can obtain the corresponding
probabilistic relation $R'$ using the rules $S(x,\distname{Flip}[p])
\dla R(x,p)$ and $R'(x) \dla S(x,1)$.  Finally, we note that a
\e{disjunctive} Datalog rule~\cite{DLV97}, where the conclusion can be
a disjunction of atoms, can be simulated by our model (with
probabilities ignored): If the conclusion has $n$ disjuncts, then we
construct a distributional rule with a probability distribution over
$\{1,\ldots,n\}$, and additional $n$ deterministic rules corresponding
to the atoms.

\section{Chasing Generative Programs}\label{sec:chase}

\emph{The chase}~\cite{MMS79,ABU79} is a classic technique used for
reasoning about tuple-generating dependencies and equality-generating
dependencies. In the special case of full tuple-generating
dependencies, which are syntactically isomorphic to Datalog rules, the
chase is closely related to (a tuple-at-a-time version of) the naive
\emph{bottom-up evaluation} strategy for Datalog program
(cf.~\cite{AbiteboulHV95}). In this section, we present a suitable
variant of the chase for generative Datalog programs, and analyze some
of its properties. The goal of that is twofold. First, as we will
show, the chase provides an intuitive executional counterpart of the
declarative semantics in Section~\ref{sec:generative}.
Second, we use the chase to prove Theorems~\ref{thm:discrete}
and~\ref{thm:at-most-one}.

We note that, although the notions and results could arguably be
phrased in terms of a probabilisitic extension of bottom-up Datalog
evaluation strategy, the fact that a \DDatalog rule can create new
values makes it more convenient to phrase them in terms of a suitable
adaptation of the chase procedure.

To simplify the notation in this
section, we fix a \DDatalog program $\G=(\E,\I,\Theta)$.  
Let $\widehat{\G} = (\E,\I^\Delta, \Theta^\Delta)$
be the associated \EDatalog program. 

We define the notions of \emph{chase step} and \e{chase tree}.

\partitle{Chase step} Consider an instance $J$, a rule
$\tau\in\Theta^\Delta$ of the form $\psi(\tup x)\leftarrow\varphi(\tup
x)$, and a tuple $\tup a$ such that $\varphi(\tup a)$ is satisfied in
$J$ but $\psi(\tup a)$ is not satisfied in $J$. If $\psi(\tup x)$ is a
distributional atom of the form $\exists y R^\delta_i(\tup t, y, \tup
t', \tup p)$, then $\psi$ being ``not satisfied'' is interpreted in
the logical sense (regardless of probabilities): there is no $y$ such
that the tuple $(\tup t, y, \tup t', \tup p)$ is in $R^\delta_i$. In
that case, let $\mathcal{J}$ be the
set of all instances $J_b$ obtained by extending $J$ with $\psi(\tup
a)$ for a specific value $b$ of the existential variable $y$, such
that $\delta(b|\tup p)>0$.  Furthermore, let $\pi$ be the discrete
probability distribution over $\mathcal{J}$ that assigns to $J_b$ the
probability mass $\delta(b|\tup p)$.  If $\psi(\tup x)$ is an ordinary
atom without existential quantifiers, $\mathcal{J}$ is simply defined
as $\{J'\}$, where $J'$ extends $J$ with the facts in $\psi(\tup a)$,
and $\pi(J')=1$. Then, we say that
\[J\xrightarrow{\tau(\tup a)} (\mathcal{J},\pi)\] is a \emph{valid
  chase step}.

\partitle{Chase tree}
Let $I$ be an input instance for $\G$.  A \emph{chase tree for $I$}
w.r.t.~$\G$ is a possibly infinite tree, whose nodes are labeled by
instances over $\E\cup\I$ and where each edge is labeled by a real
number $r\in [0,1]$ such that
\begin{enumerate}
\item The root is labeled by $I$;
\item For each non-leaf node labeled $J$, if $\mathcal{J}$ is the set
  of labels of the children of the node, and if $\pi$ is the map
  assigning to each $J'\in\mathcal{J}$ the label of the edge from $J$
  to $J'$, then $J\xrightarrow{\tau(\tup a)} (\mathcal{J},\pi)$ is a
  valid chase step for some rule $\tau\in\Theta^\Delta$ and tuple
  $\tup a$.
\item For each leaf node labeled $J$, there does not exist a valid
  chase step of the form $J\xrightarrow{\tau(\tup a)} (\mathcal{J},
  \pi)$.  In other words, the tree cannot be extended to a larger
  chase tree.
\end{enumerate}

We denote by $L(v)$ the label (instance) of the node $v$.  Each
instance $L(v)$ of a node of $v$ of a chase tree is said to be an
\emph{intermediate instance} w.r.t.~that chase tree.  A chase tree is
said to be \emph{injective} if no intermediate instance is the label
of more than one node; that is, for $v_1\neq v_2$ we have $L(v_1)\neq
L(v_2)$. As we will see shortly, due to the specific construction of
$\Theta^\Delta$, every chase tree turns out to be injective.

\medskip\par\noindent\textbf{Properties of the chase.}  We now state
some properties of our chase procedure.

\def\propfd{ Let $I$ be any input instance, and consider any chase
  tree for $I$ w.r.t.~$\G$. Then every intermediate instance satisfies
  all functional dependencies associated to distributional relations.
}
\begin{proposition}\label{prop:fd} 
 \propfd 
\end{proposition}

\def\proptree{ Every chase tree w.r.t.~$\G$ is injective.  }
\begin{proposition} \label{prop:tree}
\proptree
\end{proposition}

% We will use the notation $\brg{\G}{I}$ to refer to an arbitrary chase tree 
% of $I$ with respect to $\G$.
We denote by $\leaves(T)$ the set of leaves of a chase tree $T$, and
we denote by $L(\leaves(T))$ the set $\set{L(v)\mid v\in\leaves(T)}$.

\begin{theorem}\label{theorem:leaf-flfp}
  Let $T$ be a chase tree for an input instance $I$ w.r.t.~%a \DDatalog program 
  $\G$.
  The following hold.
  \begin{enumerate}
  \item Every intermediate instance is a subset of some possible outcome in
    $\sol_\G(I)$.
    \item If $T$ does not have infinite directed paths, then
      $L(\leaves(T))=\fsol_\G(I)$.
  \end{enumerate}
\end{theorem}
This theorem is a special case of a more general result,
Theorem~\ref{theorem:maxpaths-lfp}, which we prove later.

\subsection{Proof of Theorems~\ref{thm:discrete}
  and~\ref{thm:at-most-one}} 
By construction, for every node of a chase tree $T$, the weights of
the edges that emanate from the node in question sum up to one.  We
can associate to each intermediate instance $L(v)$ a \emph{weight},
namely the product of the edge labels on the path from the root to
$v$. This weight is well defined, since $T$ is injective.  We can then
consider a random walk over the tree, where the probabilities are
given by the edge labels.  Then, for a node $v$, the weight of $L(v)$
is equal to the probability of visiting $v$ in this random world. From
Theorem~\ref{theorem:leaf-flfp} we conclude that, if all the possible outcomes are
finite, then $T$ does not have any infinite paths, and moreover, the
random walk defines a probability distribution over the labels of the  leaves, which
are the possible outcomes. This is precisely the probability distribution of
Theorem~\ref{thm:discrete}.  Moreover, in the general case,
$\Sigma_{J\in\fsol_\G(I)}\Pr_{\G,I}(J)$ is the probability that the
random walk terminates (at a leaf), and hence,
Theorem~\ref{thm:at-most-one} follows from the fact that this
probability (as is any probability) is a number between zero and one.

%%%%%%%%%%%%%%%%%%%%%%%%%%%%%%%%%%
%%%%%%%%%%%%%%%%%%%%%%%%%%%%%%%%%%
%\input{sections/measures.tex}
\section{Infinite Possible Outcomes}\label{sec:measures}

In the general case of a \DDatalog program, possible oucomes may be infinite,
and moreover, the space of possible outcomes may be uncountable.

\begin{example}
  We now discuss examples that show what would happen if we
  straightforwardly extended our current definition of the probability
  $\Pr_{\G,I}(J)$ of possible outcomes $J$ to infinite possible
  outcomes (where, in the case where $J$ is infinite, $\Pr_{\G,I}(J)$
  would be the limit of an infinite product of weights).

  Consider the $\DDatalog$ program defined by the rule
  $R(y,\delta[y])\leftarrow R(x,y)$ where $\delta$ is a probability
  distribution with one parameter $p$ and such that $\delta(z|p)$ is
  equal to $1$ if $z=2p$ and $0$ otherwise.  Then, $I=\{R(0,1)\}$ has
  no finite possible outcome.  In fact, $I$ has exactly one infinite
  possible outcome: $\{R(0,1)\}\cup\{R^{\delta}_2(2^i,2^{i+1},2^i)
  \mid i \geq 0\}\cup\{R(2^i,2^{i+1}) \mid i \geq 0\}$.

  Now consider the previous program extended with the rule
  $R(0,\distname{Flip}[0.5])\leftarrow Q(x)$, and consider the input
  instance $I'=\{Q(0)\}$.  Then, $I'$ has one finite possible outcome
  $J = \{Q(0),R(0,0),R^{\distname{Flip}}_2(0,0,0.5)\}$ with
  $\Pr_{\G,I'}(J)=0.5$, and another infinite possible outcome
  $J'=\{R(0,1),
  R^{\distname{Flip}}_2(0,1,0.5)\}\cup\{R^{\delta}_2(2^i,2^{i+1},2^i)
  \mid i \geq 0\}\cup\{R(2^i,2^{i+1}) \mid i \geq 0\}$ with
  $\Pr_{\G,I'}(J')=0.5$. 

  Next, consider the $\DDatalog$ program defined by
  $R(y,\delta'[y])\leftarrow R(x,y)$, where $\delta'$ is a probability
  distribution with one parameter $p$, and $\delta'(z|p)$ is equal to
  $0.5$ if $z\in\{2p, 2p+1\}$ and $0$ otherwise. Then, for
  $I=\{R(0,1)\}$, every possible outcome is infinite, and would have
  the probability 0.

  Now consider the previous program extended with the rule
  $R(0,\distname{Flip}[0.5])\leftarrow Q(x)$, and consider again the
  input instance $I'=\{Q(0)\}$.  Then $I'$ would have exactly one
  possible outcome $J$ with $\Pr_{\G,I'}(J)>0$, namely $J=\{Q(0),
  R(0,0), R^\distname{Flip}_2(0,0,.5), R^{\delta'}_2(0,0,0)\}$ where
  $\Pr_{\G,I'}(J)=0.25$.
\end{example}

\subsection{Generalization of Probabilistic Semantics}
To generalize our framework, we need to consider probability spaces
over uncountable domains; those are defined by means of measure
spaces, which are defined as follows.

%\subsection{Probability Measure Spaces}
Let $\Omega$ be a set. A \e{$\sigma$-algebra} over $\Omega$ is a
collection $\F$ of subsets of $\Omega$, such that $\F$ contains
$\Omega$ and is closed under complement and countable unions. (Implied
properties include that $\F$ contains the empty set, and that $\F$ is
closed under countable intersections.) If $\F'$ is a nonempty
collection of subsets of $\Omega$, then the closure of $\F'$ under
complement and countable unions is a $\sigma$-algebra, and it is said
to be \e{generated} by $\F'$.

A \e{probability measure space} is a triple $(\Omega,\F,\pi)$, where:
%\begin{itemize}
%\item 
\e{(1)} $\Omega$ is a set, called the \e{sample space},
%\item 
\e{(2)} $\F$ is a $\sigma$-algebra over $\Omega$,
%\item 
\e{and (3)} $\pi:\F\rightarrow[0,1]$, called a \e{probability
  measure}, is such that $\pi(\Omega)=1$, and
$\pi(\cup\E)=\sum_{e\in\E}\pi(e)$ for every countable set $\E$ of
pairwise-disjoint measurable sets.
% \end{itemize}

Let $\G$ be a \DDatalog program, and let $I$ be an input for $\G$. We
say that a sequence $\tup f=(f_1,\dots,f_n)$ of facts is a
\emph{derivation} (w.r.t.~$I$) if for all $i=1,\dots,n$, the fact
$f_i$ is the result of applying some rule of $\G$ that is not
satisfied in $I\cup \set{f_1,\dots,f_{i-1}}$ (in the case of applying
a rule with a $\Delta$-atom in the head, choosing a value
randomly). If $f_1,\dots,f_n$ is a derivation, then the set
$\set{f_1,\dots,f_n}$ is a \e{derivation set}. Hence, a finite set $F$
of facts is a derivation set if and only if $I\cup F$ is an
intermediate instance in some chase tree.

Let $\G$ be a \DDatalog program, let $I$ be an input for $\G$, and let
$F$ be a set of facts. We denote by $\csol{F}_\G(I)$ the set of all
the possible outcomes $J\subseteq \sol_\G(I)$ such that $F\subseteq J$. The
following theorem states how we determine the measure space defined by
a \DDatalog program.

\begin{theorem}\label{thm:meausre}
  Let $\G$ be a \DDatalog program, and let $I$ be an input for
  $\G$. There exists a \e{unique} probability measure space
  $(\Omega,\F,\pi)$, denoted $\mu_{\G,I}$, that satisfies all of the
  following.
    \begin{enumerate}
    \item[(i)] $\Omega=\sol_\G(I)$;
    \item[(ii)] The $\sigma$-algebra $(\Omega,\F)$ is generated from the
      sets of the form $\csol{F}_\G(I)$ where $F$ is finite;
    \item[(iii)] $\pi(\csol{F}_\G(I)) = \prw(F)$ for every derivation
      set $F$.
    \end{enumerate}
    Moreover, if $J$ is a finite possible outcome, then
    $\pi(\set{J})$ is equal to $\prw(F)$. \end{theorem}

  Observe that the items (i) and (ii) of Theorem~\ref{thm:meausre}
  describe the unique properties of the probability measure space. The
  proof will be given in the next section. The last part of the
  theorem states that our discrete and continuous probability
  definitions coincide on finite possible outcomes; this is a simple
  consequence of item (ii), since for a finite possible outcome $J$, the set
  $F=J\setminus I$ is such that $\csol{F}_\G(I)=\set{J}$, and $F$ is
  itself a derivation set (e.g., due to
  Theorem~\ref{theorem:leaf-flfp}).

\subsection{Measure Spaces by Infinite Chase}
We prove Theorem~\ref{thm:meausre} by defining and investigating
measure spaces that are defined in terms of the chase.  Consider a
\DDatalog program $\G$ and an input $I$ for $\G$.  A \e{maximal path}
of a chase tree $T$ is a path $P$ that starts with the root, and
either ends in a leaf or is infinite.  Observe that the labels
(instances) along a maximal path form a chain (w.r.t.~the
set-containment partial order).  A maximal path $P$ of a chase tree is
\e{fair} if whenever the premise of a rule is satisfied by some tuple
in some intermediate instance on $P$, then the conclusion of the rule
is satisfied for the same tuple in some intermediate instance on $P$.
A chase tree $T$ is \e{fair} (or has the \e{fairness} property) if
every maximal path is fair.  Note that every finite chase tree is
fair.  We will restrict attention to fair chase trees.  Fairness is a
classic notion in the study of infinite computations; moreover, fair
chase trees can easily be constructed, for examples, by maintaining a
queue of ``active rule firings'' (cf.~any textbook on term rewriting
systems or lambda calculus).

Let $\G$ be a \DDatalog program, let $I$ be an input for $\G$, and let
$T$ be a chase tree. We denote by $\maxpaths(T)$ the set of all the
maximal paths of $T$. (Note that $\maxpaths(T)$ may be uncountably
infinite.) For $P\in\maxpaths(T)$, we denote by $\cup P$ the union of
the (chain of) labels $L(v)$ along $P$. The following generalizes
Theorem~\ref{theorem:leaf-flfp}.

\def\theoremmaxpathslfp{ Let $\G$ be a \DDatalog program, $I$ an input
  for $\G$, and $T$ a fair chase tree. The mapping $P\rightarrow\cup
  P$ is a bijection between $\maxpaths(T)$ and $\sol_\G(I)$.  }

\begin{theorem}\label{theorem:maxpaths-lfp}
 \theoremmaxpathslfp
\end{theorem}

\subsection{Chase Measures}
Let $\G$ be a \DDatalog program, let $I$ be an input for $\G$, and let
$T$ be a chase tree. Our goal is to define a probability measure over
$\sol_\G(I)$. Given Theorem~\ref{theorem:maxpaths-lfp}, we can do that
by defining a probability measure over $\maxpaths(T)$. A random path
in $\maxpaths(T)$ can be viewed as a \e{Markov chain} that is defined
by a random walk over $T$, starting from the root. A measure space for
such a Markov chain is defined by means of
\e{cylinderification}~\cite{Ash2000}. Let $v$ be a node of $T$. The
\e{$v$-cylinder} of $T$, denoted $C^T_v$, is the subset of
$\maxpaths(T)$ that consists of all the maximal paths that contain
$v$. A \e{cylinder} of $T$ is a subset of $\maxpaths(T)$ that forms a
$v$-cylinder for some node $v$. 
We denote by $C(T)$ the set of all the
cylinders of $T$.  \looseness=-1

Recall that $L(v)$ is a finite set of facts, and observe that
$\prw(L(v))$ is the product of the weights along the path from the
root to $v$.  The following theorem is a special case of a classic
result on Markov chains (cf.~\cite{Ash2000}).

\begin{theorem}\label{theorem:ash}
  Let $\G$ be a \DDatalog program, let $I$ be an input for $\G$, and
  let $T$ be a chase tree.  There exists a unique probability measure
  $(\Omega,\F,\pi)$ that satisfies all of the following.
\begin{enumerate}
\item $\Omega=\maxpaths(T)$.
\item $(\Omega,\F)$ is the $\sigma$-algebra generated from $C(T)$.
\item $\pi(C^T_v)=\prw(L(v))$ for all nodes $v$ of $T$.
\end{enumerate}
\end{theorem}
Theorems~\ref{theorem:maxpaths-lfp} and~\ref{theorem:ash} suggest the
following definition.

\begin{definition}[Chase Probability Measure] \label{def:muT}
  Let $\G$ be a \DDatalog program, let $I$ be an input for $\G$, 
  let $T$ be a chase tree, and let $(\Omega,\F,\pi)$ be the
  probability measure of Theorem~\ref{theorem:ash}.  The probability
  measure $\mu_T$ over $\sol_\G(I)$ is the one obtained from
  $(\Omega,\F,\mu)$ by replacing every maximal path $P$ with the possible outcome
  $\cup P$.
\end{definition}

Next, we prove that the probability measure space represented by a
chase tree is independent of the specific chase tree of choice.  For
that, we need some notation and a lemma.  Let $\G$ be a \DDatalog
program, let $I$ be an input for $\G$, let $T$ be a chase tree, and
let $v$ be a node of $T$. We denote by $\cup C^T_v$ the set $\set{\cup
  P\mid P\in C^T_v}$. The following lemma is a consequence of
Proposition~\ref{prop:fd} and Theorem~\ref{theorem:maxpaths-lfp}.

\begin{lemma}\label{lemma:cylinder-supersets}
  Let $\G$ be a \DDatalog program, let $I$ be an input for $\G$, and
  let $T$ be a fair chase tree. Let $v$ be a node of $T$ and
  $F=L(v)$. Then $\cup C^T_v=\csol{F}_\G(I)$; that is, $\cup C^T_v$ is the
  set $\set{J\in \sol_\G(I)\mid L(v)\subseteq J}$. \end{lemma}

Using Lemma~\ref{lemma:cylinder-supersets} we can prove the following
theorem.

\def\thmchaseindependent{
Let $\G$ be a \DDatalog
  program, let $I$ be an input for $\G$, and let $T$ and $T'$ be two
  fair chase trees. Then $\mu_T=\mu_{T'}$.
}
\begin{theorem}\label{thm:chase-independent} 
\thmchaseindependent
\end{theorem}

\subsection{Proof of Theorem~\ref{thm:meausre}}

We can now prove Theorem~\ref{thm:meausre}. Let $\G$ be a \DDatalog
program, let $I$ be an input for $\G$, and let $T$ be a fair chase
tree for $I$ w.r.t.~$\G$.  Let $\mu_T=(\sol_\G(I),\F_T,\pi_T)$ be the
probability measure on $\sol_\G(I)$ associated to $T$, as defined in
Definition~\ref{def:muT}.

\begin{lemma}\label{lemma:same-sigma-algebra} 
  The $\sigma$-algebra $(\sol_\G(I),\F_T)$ is generated by the sets
  of the form $\csol{F}_\G(I)$, where $F$ is finite. 
\end{lemma} 

\begin{proof} 
  Let $(\sol_\G(I),\F)$ be the $\sigma$-algebra generated from
  the sets $\csol{F}_\G(I)$. We
  will show that every $\csol{F}_\G(I)$ is in $\F_T$, and that every $\cup
  C^{T}_{v}$ is in $\F$. The second claim is due to
  Lemma~\ref{lemma:cylinder-supersets}, so we will prove the first.
  So, let $\csol{F}_\G(I)$ be given. Due to
  Lemma~\ref{lemma:cylinder-supersets}, the set $\csol{F}_\G(I)$ is the countable union 
  $\cup_{u\in U}(\cup C^{T}_{u})$ where $U$ is the set of all the
  nodes $u$ such that $F\subseteq u$. Hence, $\csol{F}_\G(I)\in\F_T$.
 \end{proof}

\begin{lemma}\label{lemma:same-measure}
  For every derivation set $F$ we have $\pi_T(\csol{F}_\G(I)) =
  \prw(F)$.
\end{lemma} 
\begin{proof} Let $F$ be a derivation set. Due to
  Theorem~\ref{thm:chase-independent}, it suffices to prove that for
  \e{some} chase tree $T'$ it is the case that
  $\pi_{T'}(\csol{F}_\G(I))=\prw(F)$. But since $F$ is a derivation
  set, we can craft a chase tree $T'$ that has a node $v$ with
  $L(v)=F$. Then we have that $\pi_{T'}(\csol{F}_\G(I))$ is the
  product of the weights along the path to $v$, which is exactly
  $\prw(F)$.
 \end{proof}

\begin{lemma}\label{lemma:same-prob}
  Let $\mu=(\Omega,\F,\pi)$ be any probability space that satisfies
  (i)--(iii) of Theorem~\ref{thm:meausre}. Then $\mu=\mu_T$.
\end{lemma} 
\begin{proof} Let $\mu_T=(\sol_\G(I),\F_T,\pi_T)$. Due to
  Lemma~\ref{lemma:same-sigma-algebra}, we have that $\F=\F_T$.  So it
  is left to prove that $\pi=\pi_T$. Due to
  Lemmas~\ref{lemma:same-measure} and~\ref{lemma:cylinder-supersets},
  we have that $\pi$ agrees with $\pi_T$ on the cylinder sets of
  $T$. Due to Theorems~\ref{theorem:maxpaths-lfp}
  and~\ref{theorem:ash} we get that $\pi$ must be equal to $\pi_T$ due
  to the uniqueness of $\pi_T$.
 \end{proof}

The above lemmas show that $\mu_T=(\sol_\G(I),\F_T,\pi_T)$ is a probability
measure space that satisfies (i) and (ii) of Theorem~\ref{thm:meausre}, and
moreover, that no other probability measure space satisfies (i) and
(ii).

%%%%%%%%%%%%%%%%%%%%%%%%%%%%%%%%%%
%%%%%%%%%%%%%%%%%%%%%%%%%%%%%%%%%%
%\input{sections/ppdl.tex}
\newcommand{\ppdl}{\text{PPDL\textup{[}$\Delta$\textup{]}}\xspace}
\section{Probabilistic-Programming Datalog}\label{sec:ppdl}

To complete our framework, we define \e{probabilistic-programming
  Datalog}, \e{PPDL} for short, wherein a program augments a
generative Datalog program with constraints; these constraints unify
the traditional \e{integrity constraints} of databases and the
traditional \e{observations} of probabilistic programming.

\begin{definition}[\ppdl]
  Let $\Delta$ be a finite set of parametrized numerical
  distributions. A \e{\ppdl program} is a quadruple
  $(\E,\I,\Theta,\Phi)$, where $(\E,\I,\Theta)$ is a $\DDatalog$
  program and $\Phi$ is a finite set of logical constraints over
  $\E\cup\I$.\footnote{The restriction of the language of constraints
    to a fragment with tractability (or other goodness) properties is
    beyond the scope of this paper.}
\end{definition}

\begin{example}\label{example:ppdl}
  Consider again Example~\ref{example:generative}. Suppose that we
  have the EDB relations $\rel{ReportHAlarm}$ and
  $\rel{ReportBAlarm}$ that represent reported home and business
  alarms, respectively. We obtain from the program in the example a
  $\ppdl$-program by adding the following constraints.
\begin{enumerate}
  \item $\rel{ReportHAlarm}(h)\rightarrow \rel{Alarm}(h)$
  \item $\rel{ReportBAlarm}(b)\rightarrow \rel{Alarm}(b)$
  \end{enumerate}
  Note that we use right (in contrast to left) arrows to distinguish
  constraints from ordinary Datalog rules.
\end{example}

The semantics of a \ppdl program is the posterior distribution over
the corresponding $\DDatalog$ program, conditioned on the satisfaction
of the constraint. A formal definition follows.

Let $\P=(\E,\I,\Theta,\Phi)$ be a \ppdl program, and let $\G$ be the
\DDatalog program $(\E,\I,\Theta)$. An \e{input instance} for $\P$ is
an input instance $I$ for $\G$.  We say that $I$ is a \e{legal} input
instance for $\P$ if $\{J\in\sol_{\G}(I)\mid J\models\Phi\}$ is a
measurable set in the probability space $\mu_{\G,I}$, and moreover,
its measure is nonzero. Intuitively, an input instance $I$ is legal if
it is consistent with the observations (i.e., with the conjunction of
the constraints in $\Phi$), given $\G$.  The semantics of a $\ppdl$
program is defined as follows.

\begin{definition}
  Let $\P=(\E,\I,\Theta,\Phi)$ be a \ppdl program, and let $\G$ be the
  \DDatalog program $(\E,\I,\Theta)$. Let $I$ be a legal input instance for $\P$,
and let $\mu_{\G,I}=(\sol_\G(I),\F_{\G},\pi_{\G})$. The probability space defined by
  $\P$ and $I$, denoted $\mu_{\P,I}$, is the triple
    $(\sol_\P(I),\F_\P,\pi_\P)$ where:
    \begin{itemize}
      \item $\sol_\P(I) = \{J\in\sol_\G(I)\mid J\models\Phi\}$
      \item $\F_\P=\set{S\cap \sol_\P(I)\mid S\in\F_{\G}}$
      \item $\pi_\P(S)=\pi_{\G}(S)/\pi_{\G}(\sol_\P(I))$ for every
        $S\in\F_\P$.
        \end{itemize}
   In other words,  $\mu_{\P,I}$  is %the probability space
  $\mu_{\G,I}$ conditioned on $\Phi$.
\end{definition}

% Let $\P=(\E,\I,\Theta,\Phi)$ be a \ppdl program, and let $\G$ be the
% \DDatalog program $(\E,\I,\Theta)$. An \e{input instance} for $\P$ is
% an input instance $I$ for $\G$. A \e{possible outcome} of an input
% instance $I$ w.r.t.~$\P$ is a possible outcome $J$ w.r.t.~$\G$, such
% that $J$ satisfies $\Phi$. We denote by $\ppdlsol_{\Theta|\Phi}(I)$
% the set of possible outcomes for $I$ w.r.t.~$\P$. We say that $I$ is a
% \e{legal} input instance for $\P$ if $\ppdlsol_{\Theta|\Phi}(I)$ is a
% measurable set in the probability space $\mu_{\G,I}$, and moreover,
% its measure is nonzero. Intuitively, an input instance $I$ is legal
% if it is consistent with the observations (i.e., with $\Phi$) w.r.t.~the given
% $\DDatalog$ program.
%  The semantics of a $\ppdl$ program is defined
% next.

% \begin{definition}
%   Let $\P=(\E,\I,\Theta,\Phi)$ be a \ppdl program, and let $\G$ be the
%   \DDatalog program $(\E,\I,\Theta)$.
%    Let $\mu_{\G,I}=(\sol_\G(I),\F_{\G},\pi_{\G})$. The probability space defined by
%   $\P$ and $I$, denoted $\mu_{\P,I}$, is the probability space
%   $\mu_{\G,I}$ conditioned on $\Phi$; that is, it is the triple
%     $(\ppdlsol_{\Theta|\Phi}(I),\F,\pi)$ where:
%     \begin{itemize}
%       \item $\F=\set{S\cap \ppdlsol_{\Theta|\Phi}(I)\mid S\in\F_{\G}}$
%       \item $\pi(S)=\pi_{\G}(S)/\pi_{\G}(\ppdlsol_{\Theta|\Phi}(I))$ for every
%         $S\in\F$.
%         \end{itemize}
% \end{definition}

\begin{example}
  Continuing Example~\ref{example:ppdl}, the semantics of this program
  is the posterior probability distribution for the prior of
  Example~\ref{example:generative}, under the conditions that the
  alarm is on whenever it is reported. Then, one can ask various
  queries over the probability space defined, for example, the
  probability of the fact ${Earthquake}(\val{Napa},1)$.  Observe
  that, with negation, we could also phrase the condition that an
  alarm is off unless reported.
\end{example}

We note that a traditional integrity constraint on the input (e.g.,
there are no $x$, $c_1$ and $c_2$ such that both $\rel{Home}(x,c_1)$
and $\rel{Business}(x,c_2)$ hold) can be viewed as a constraint in
$\Phi$ that holds with either probability $0$ (and then the input is
illegal) or with probability $1$ (and then the prior $\mu_{\G,I}$ is
the same as the posterior $\mu_{\P,I}$).

An important direction for future work is to establish tractable
conditions that guarantee that a given input is legal. Also, an
interesting problem is to detect conditions under which the chase is a
\e{self conjugate}~\cite{Raiffa+61}, that is, the probability space
$\mu_{\P,I}$ is captured by a chase procedure without backtracking.

%%%%%%%%%%%%%%%%%%%%%%%%%%%%%%%%%%
%%%%%%%%%%%%%%%%%%%%%%%%%%%%%%%%%%
%\input{sections/extensions.tex}
\section{Extensions and Future Work}\label{sec:extensions}

Our ultimate goal is to design a language for probabilistic
programming that possesses the inherent declarative and logical nature
of Datalog. To that end, extensions are required. In this section we
discuss some of the important future directions and challenges to
pursue. We focus on the semantic aspects of expressive power. (An
obvious aspect for future work is a practical implementation, e.g.,
corresponding sampling techniques.)

\subsection{Unbounded Number of Parameters}
It is often the case that the probability distributions have a large
number of parameters. A simple example is the \e{categorical}
distribution where a single member of a finite domain of items is to
be selected, each item with its own probability.  In that case, the
domain can be very large, and moreover, it can be determined
dynamically from the database content and be unknown to the static
program.  To support such distributions, one can associate the
distribution with a relation symbol in a schema, and here we
illustrate it for the case of a categorical distribution.

Let $R$ be a relation symbol. A \e{categorical distribution over $R$}
is associated with two attributes $i$ and $j$ of $R$ that determine
\e{dynamic parameters}: $i$ represents a possible value, and $j$
represents its probability. In addition, the distribution is
associated with a tuple $\tup g$ of attributes of $R$ that split $R$
into \e{groups} with the semantics of SQL's GROUP BY. We denote this
distribution by $\distname{Cat}^R\langle{i,j;\tup g}\rangle(\tup
x;\tup q)$, where $\tup g$ and $\tup q$ have the same length. Given a
relation $R^I$ over $R$ and parameters $\tup q$, let $R^I_{\tup g}$ be
the sub-relation of $R^I$ that have the values in the attributes
vector $\tup g$ are equal to $\tup q$. Suppose that the facts of
$R^I_{\tup q}$ are $f_1,\dots,f_n$. We assume that $f_k[j]\in[0,1]$
for all $k=1,\dots,n$, and moreover, that $\sum_{k=1}^nf_k[j]=1$. Then
we define $\distname{Cat}^R\langle{i,j|\tup g}\rangle(x | \tup q)=s$,
where $s$ is the sum of $f_k[j]$ over all $k\in\set{1,\dots,n}$ with
$f_k[i]=x$.

As an example, consider the relation $\rel{Cor}(w_c,w_e,q)$ that
provides, for every English word $w_c$, a distribution over the
possible misspelled words $w_e$; hence, the fact
$\rel{Cor}(w_c,w_e,q)$ means that $w_c$ is misspelled into $w_e$ with
probability $q$. In our notation, this distribution will be captured
by the notation
$\distname{Cat}^{\rel{Cor}}\langle{2,3;1}\rangle(w_e;w_c)$. Hence, the
following program contains, for each document, the set of words that
the document can contain by replacing each word with a corresponding
correction. The relation $\rel{Doc}(d,i,w)$ denotes that document $d$
has the word $w$ as its $i$th token. The relation $\rel{CDoc}$ is the
same as $\rel{Doc}$, except that each word is replaced with a random
correction. \begin{align}\label{eq:cat}
  \rel{CDoc}(d,i,\distname{Cat}^{\rel{Cor}}\langle{2,3;1}\rangle(w_c;w))
  \dla \rel{Doc}(d,i;w)
\end{align}

The above specification for the categorical distribution is similar to
the \verb=repair-key= operation of the world-set
algebra~\cite{WSA07,PWSA07}.  In the general case, we define an
\e{$R$-distribution} to be one of the form $\delta^R\langle{\tup
  d;\tup g\rangle}(x;\tup q)$, where $R$ is a relation symbol, $\tup
d$ is a tuple of attributes of $R$, representing the dynamic
parameters (i.e., separate parameters for every row), $\tup g$ is a
tuple of grouping attributes, and $\tup q$ has the same length as
$\tup g$. We could also add a tuple $\tup p$ of \e{static parameters}
(i.e., ones that are shared by all the tuples in the group).

Extending our semantics to support a program such as~\eqref{eq:cat} is
straightforward, since the relation $\rel{Cor}$ that defines the
distribution is an EDB relation. The only difference from our previous
definitions is that, now, the probability of a distributional fact is
determined not only by the fact itself, but also by the content of the
relation to which the fact refers to (e.g., $\rel{Cor}$). When the
relation $R$ that is used for defining the distribution is an IDB, we
need to be careful with the definition of the chase, since when we
wish to sample from an $R$-distribution, it may be the case that some
of the relevant facts of $R$ have not been produced yet.  We may even
be in a cyclic situation where the facts of $R$ are determined by
facts derived (indirectly) from facts that are produced using
$R$-distributions. We plan to devise a (testable) acyclicity condition
that avoids such cycles, and then restrict the chase to behave
accordingly; $R$-distributions are sampled from only after $R$ is
complete. Of course, it would be interesting to explore the semantics
of programs without the acyclicity property.

\subsection{Multivariate Distributions}
A natural and important extension is to support \e{multivariate
  distributions}, which are distributions with a support in
$\allreals^k$ for $k>1$. Examples of popular such distributions are
multinomial, Dirichlet, and multivariate Gaussian distribution.  When
$k$ is fixed, one can replace our single distributional term with
multiple such terms. But when $k$ is unbounded, such a distribution
should be supported as an \e{aggregate} operation that implies in a
\e{set} of facts (rather than a single one).

\subsection{Continuous Distributions}

A natural extension of our framework is the support of continuous
probability distributions (e.g., continuous uniform, Pareto, Gaussian,
Dirichlet, etc.).  This is a very important extension, as such
distributions are highly popular in defining statistical
models. Syntactically, this extension is straightforward: we just need
to include these distributions in $\Delta$. Likewise, extending the
probabilistic chase is also straightforward.

The challenge, though, is with the semantic analysis, and in
particular, with the definition of the probability space implied by
the chase. When a chase step involves a continuous numeric
distribution, such as $U(0,1)$ (the uniform distribution between $0$
and $1$), then no chase step is measurable, and hence, we can no
longer talk about the probability of a step or the probability of a
cylinder set (but we can talk about the \e{density} of those).  Note
that our definition of the measure space in Section~\ref{sec:measures}
is inherently based on the assumption that the set of possible outcomes that
contains a given finite set of facts is measurable.  But to support
continuous distribution, the definition of measurable sets will need
to be based on sets of paths. We refer this to future work, and we
believe that our current framework will naturally extend to continuous
distributions.

%%%%%%%%%%%%%%%%%%%%%%%%%%%%%%%%%%
%%%%%%%%%%%%%%%%%%%%%%%%%%%%%%%%%%
%\input{sections/conclusions.tex}
\section{Concluding Remarks}\label{sec:conclusions}

We proposed and investigated a declarative framework for specifying
statistical models in the context of a database, based on an extension
of Datalog with numerical distributions. The framework differs from
traditional probabilistic programming languages not only due to the
tight integration with a database, but also because of its fully
declarative rule-based language: the interpretation of a program is
independent under transformations (such as reordering or duplication
of rules) that preserve the first-order semantics. This was achieved
by treating a \DDatalog program as a Datalog program with
existentially quantified variables in the conclusion, and using the
minimal solutions as the sample space of a (discrete or continuous)
probability distribution.  Using a suitable notion of chase that we
introduced, we established that the resulting probability
distributions are well-defined and robust.

This work is done as part of the effort to extend the LogicBlox
database~\cite{LB}, and its Datalog-based data management language
LogiQL~\cite{LogiQL14}, to support the specification of statistical
models. Through its purely declarative rule-based syntax, such an
extension of LogiQL allows for natural specifications of statistical
models. Moreover, there is a rich literature on (extensions of)
Datalog, and we expect that, through our framework, techniques and
insights from this active research area can be put to use to advance
the state of the art in probabilistic programming.

%%%%%%%%%%%%%%%%%%%%%%%%%%%%%%%%%%
%%%%%%%%%%%%%%%%%%%%%%%%%%%%%%%%%%
% A category with the (minimum) three required fields
%\category{H.4}{Information Systems Applications}{Miscellaneous}
%A category including the fourth, optional field follows...
%\category{D.2.8}{Software Engineering}{Metrics}[complexity measures, performance measures]

%\terms{Delphi theory}

%\keywords{ACM proceedings, \LaTeX, text tagging}

\section*{Acknowledgments}
We are thankful to Molham Aref, Todd J.~Green and Emir Pasalic for
insightful discussions and feedback on this work.  We also thank
Michael Benedikt, Georg Gottlob and Yannis Kassios for providing
useful comments and suggestions. Finally, we are grateful to Kathleen
Fisher and Suresh Jagannathan for including us in DARPA's PPAML
initiative; this work came from our efforts to design a principled
approach to translating probabilistic programs into statistical
solvers.

{
%\small
%\scriptsize
\bibliographystyle{abbrv}
\bibliography{ppdl-arxiv}  
}
%\newpage
\appendix

%%%%%%%%%%%%%%%%%%%%%%%%%%%%%%%%%
%%%%%%%%%%%%%%%%%%%%%%%%%%%%%%%%%
%\input{sections/proofs.tex}
\section{Additional Proofs}

\def\refthemweaklyacyclic{\ref{thm:weakly-acyclic}}
\subsection{Proof of Theorem~\refthemweaklyacyclic}
\begin{reptheorem}{\ref{thm:weakly-acyclic}}
\thmweaklyacyclic
\end{reptheorem}

\begin{proof}
It is easy to show that if $\G$ is weakly acyclic in our sense, 
then $\widehat{P}$ is weakly acyclic according to the classic 
definition of weak ayclicity given in \cite{DBLP:conf/icdt/FaginKMP03}. 
We then apply the following result (restated here to match our notation), 
which was established in \cite{DBLP:journals/tods/FuxmanKMT06}:
if a \EDatalog program $\D$ is weakly acyclic, then there is a 
polynomial $p(\cdot)$ (depending only on $\D$) such that 
for every input instance $I$ and for every solution $J$
of $I$ w.r.t.~$\D$, there is a solution $J'$ of $I$ w.r.t.~$\D$
with $J'\subseteq J$ and $|J'|\leq p(|I|)$. In particular, 
 all $J\in \lfp_\D(I)$ are finite and have size
at most $p(|I|)$. 
\end{proof}

\def\refpropfd{\ref{prop:fd}}
\subsection{Proof of Proposition~\refpropfd}
\begin{repproposition}{\ref{prop:fd}}
\propfd
\end{repproposition}
\begin{proof} The proof proceeds by induction on the distance from the
  root of the chase tree. Suppose that in a chase step
  $J\xrightarrow{\tau(\tup a)} (\mathcal{J},\pi)$, some
  $J'\in\mathcal{J}$ contains two $R^\delta_i$-facts that are
  identical except for the $i$-th attribute. Then either $J$ already
  contains both atoms, in which case we can apply our induction
  hypothesis, or $J'$ is obtained by extending $J$ with one of the two
  facts in question, in which case, it is easy to see that the
  conclusion of $\tau$ was already satisfied for the tuple $\tup a$,
  which is not possible in case of a valid chase step.
\end{proof}

\def\refproptree{\ref{prop:tree}}
\subsection{Proof of Proposition~\refproptree}

\begin{repproposition}{\ref{prop:tree}}
\proptree
\end{repproposition}

\begin{proof}
  For the sake of a contradiction, assume that two nodes $n_1$ and
  $n_2$ in a chase tree are labeled by the same instance $J$.  Let
  $n_0$ be the node that is the least common ancestor of $n_1$ and
  $n_2$ in the tree, and let $n'_1$ and $n'_2$ be the children of
  $n_0$ that are ancestors of $n_1$ and $n_2$, respectively. By
  construction, $n'_1$ and $n'_2$ are labeled with distinct instances
  $J_1 \neq J_2$, respectively. Consider the rule $\tau = \psi(\tup
  x)\leftarrow \varphi(\tup x)$ and tuple $\tup a$ constituting the
  chase step applied at node $n_0$.  Since $n_0$ has more than one
  child, $\psi(\tup x)$ must be a distributional atom, say $\exists y
  R^\delta_i(\tup t, y, \tup t', \tup p)$.  Then each $J_k$ ($k=1,2$)
  contains an $R^\delta_i$-fact. Moreover, the two $R^\delta_i$-facts
  in question differ in the choice of value for the variable $y$, and
  are otherwise identical. Due to the monotonic nature of the chase,
  both atoms must belong $J$, and hence, $J'$ violates the functional
  dependency of Proposition~\ref{prop:fd}.  Hence, we have reached a
  contradiction.
\end{proof}

\def\reftheoremmaxpathslfp{\ref{theorem:maxpaths-lfp}}
\subsection{Proof of Theorem~\reftheoremmaxpathslfp}
\begin{reptheorem}{\ref{theorem:maxpaths-lfp}}
 \theoremmaxpathslfp
\end{reptheorem}
\begin{proof}
  We first prove that every $\cup P$ is in $\sol_\G(I)$.  Let
  $P\in\maxpaths(T)$ be given. We need to show that $\cup
  P\in\sol_\G(I)$.  By definition it is the case that every
  distributional fact of $\cup P$ has a nonzero probability. It is
  also clear that $\cup P$ is consistent, due to the fairness property
  of $T$. Hence, it suffices to prove that $\cup P$ is a \e{minimal}
  solution, that is, no proper subset of $\cup P$ is
  a solution. So, let $K$ be a strict subset of $\cup P$ and suppose,
  by way of contraction, that $K$ is also a solution. Let $(J,J')$
  be the first edge in $P$ such that $\cup P$ contains a fact that is
  not in $K$. Now, consider the chase step that leads from $J$ to
  $J'$.  Let $f$ be the unique fact in $J'\setminus J$. Then
  $J\subseteq K$ and $f\in J'\setminus K$. The selected rule $\tau$ in
  this step cannot be deterministic, or otherwise $K$ must contain $f$
  as well. Hence, it is a distributional rule, and $f$ has the form
  $R^\delta_i(\tup a| \tup p)$. But then, $K$ satisfies this rule, and
  hence, $K$ must include a fact $f'=R^\delta_i(\tup a'| \tup p)$,
  where $\tup a'$ differs from $\tup a$ only in the $i$th element. And
  since some node in $\cup P$ contains both $f$ and $f'$, we get a
  violation of the fd of Proposition~\ref{prop:fd}. Hence, a
  contraction.
  
  Next, we prove that every possible outcome $J$ in $\sol_\G(I)$ is equal to
  $\cup P$ for some $P\in\maxpaths(T)$.  Let such $J$ be given. We
  build the path $P$ inductively, as follows. We start with the root,
  and upon every node $v$ we select the next edge to be one that leads
  to a subset $K$ of $J$; note that $K$ must exist since $J$ resolves
  the rule violated in $L(v)$ by some fact, and that fact must be in
  one of the children of $v$. Now, $\cup P$ is consistent since $T$ is
  fair, and $\cup P\subseteq J$ by construction. And since $J$ is a
  minimal solution, we get that $\cup P$ is in fact equal to $J$.
  
  Finally, we need to prove that if $\cup P_1=\cup P_2$ then
  $P_1=P_2$. We will prove the contrapositive statement. Suppose that
  $P_1,P_2\in\maxpaths(T)$ are such that $P_1\neq P_2$. The two paths
  agree on the root. Let $J$ be the first node in the paths such that
  the two paths disagree on the outgoing edge of $J$. Suppose that
  $P_1$ has the edge from $J$ to $J_1$ and $P_2$ has an edge from $J$
  to $J_2$. Then $J_1\cup J_2$ have a pair of facts that violate the
  functional dependency of Proposition~\ref{prop:fd}, and in
  particular, $J_1\not\subseteq\cup P_2$. We conclude that $\cup
  P_1=\cup P_2$, as claimed.
\end{proof}

\def\refthmchaseindependent{\ref{thm:chase-independent}}
\subsection{Proof of Theorem~\refthmchaseindependent}
\begin{reptheorem}{\ref{thm:chase-independent}}
\thmchaseindependent
\end{reptheorem}
\begin{proof}
  Let $\mu_T=(\Omega,\F,\pi)$ and $\mu_{T'}=(\Omega',\F',\pi')$.  We
  need to prove that $\Omega=\Omega'$, $\F=\F'$ and $\pi=\pi'$. We
  have $\Omega=\Omega'$ due to Theorem~\ref{theorem:maxpaths-lfp}.  To
  prove that $\F=\F'$, it suffices to prove that every $\cup C^T_v$ is
  $\F'$ and every $\cup C^{T'}_{v'}$ is in $\F$ (since both
  $\sigma$-algebras are generated by the cylinders). And due to
  symmetry, it suffices to prove that $C^{T'}_{v'}$ is in $\F$.  So,
  let $v'$ be a node of $T'$. Recall that $L(v')$ is a set of facts.
  Due to Lemma~\ref{lemma:cylinder-supersets}, we have that $\cup
  C^{T'}_{v'}$ is precisely the set of all possible outcomes $J$ in
  $\sol_\G(I)$ such that $L(v')\subseteq J$. Let $U$ be the set of all
  the nodes of $u$ of $P$ with $L(v')\subseteq L(u)$.  Then, due to
  Theorem~\ref{theorem:maxpaths-lfp} we have that $\cup
  C^{T'}_{v'}=\cup_{u\in U}(\cup C^T_u)$. Observe that $U$ is
  countable, since $T$ has only a countable number of nodes (as every
  node is identified by a finite path from the root). Moreover,
  $(\Omega,\F)$ is a closed under countable unions, and therefore,
  $\cup_{u\in U}(\cup C^T_u)$ is in $\F$.

  It remains to prove that $\pi=\pi'$. By now we know that the
  $\sigma$-algebras $(\Omega,\F)$ and $(\Omega',\F')$ are the same.
  Due to Theorem~\ref{theorem:maxpaths-lfp}, every measure space over
  $(\Omega,\F)$ can be translated into a measure space over the
  cylinder algebra of $T$ and $T'$. So, due to the uniqueness property
  of Theorem~\ref{theorem:ash}, it suffices to prove that every $\cup
  C^{T'}_{v'}$ has the same probability in $\mu_T$ and
  $\mu_{T'}$. That is, $\pi(\cup C^{T'}_{v'})=\prw(L(v'))$. We do so
  next. We assume that $v'$ is not the root of $T'$, or otherwise the
  claim is straightforward.  Let $U$ be the set of all the nodes $u$
  in $T$ with the property that $L(v')\subseteq L(u)$ but
  $L(v')\not\subseteq L(p)$ for the parent $p$ of $u$. Due to
  Lemma~\ref{lemma:cylinder-supersets} we have the following:
  \begin{equation}\label{eq:sum-of-paths}
    \pi(\cup C^{T'}_{v'})=\sum_{u\in U}\prw(L(u))
    \end{equation}
    Let $E$ be the set of all the edges $(v_1,u_1)$ of $T$, such that
    $L(u_1)\setminus L(v_1)$ consists of a node in $v'$.  Let $Q$ be
    the set of all the paths from the root of $T$ to nodes in $U$.
    Due to Proposition~\ref{prop:fd}, we have that every two paths
    $P_1$ and $P_2$ in $Q$ and edges $(v_1,u_1)$ and $(v_2,u_2)$ in
    $P_1$ and $P_2$, respectively, if both edges are in $E$ and
    $v_1=v_2$, then $u_1=u_2$.  Let $T''$ be the tree that is obtained
    from $T$ by considering every edge $(v_1,u_1)$ in $E$, changing
    its weight to $1$, and changing the weights of remaining
    $(v_1,u'_1)$ emanating from $v_1$ to $0$. Then we have the
    following for every node $u\in U$.
    \begin{equation}\label{eq:product-w}
      \prw(u)= w_{T''}(u)\cdot \prw(L(v'))
      \end{equation}
      where $w_{T''}(u)$ is the product of the weights along the 
      path from the root of $T''$ to $u$.
        Combining \eqref{eq:sum-of-paths} and \eqref{eq:product-w}, we get the following.
        \[\pi(\cup C^{T'}_{v'})=\prw(L(v'))\cdot\sum_{u\in U} w_{T''}(u)\]  
        Let $p=\sum_{u\in U} w_{T''}(u)$. We need to prove that
        $p=1$. Observe that $p$ is the probability of visiting a node
        of $U$ in a random walk over $T''$ (with the probabilities
        defined by the weights). Equivalently, $p$ is the probability
        that random walk over $T''$ eventually sees all of the facts
        in $v'$. But due to the construction of $T''$, every rule
        violation that arises due to facts in both $L(v')$ and any
        node of $T''$ is deterministically resolved exactly as in
        $L(v')$. Moreover, since $L(v')$ is obtained from a chase
        derivation (i.e., $L(v')$ is a derivation set), solving all
        such rules repeatedly results in the containment of
        $L(v')$. Finally, since $T''$ is fair (because $T$ is fair),
        we get that every random walk over $T''$ eventually sees all
        of the facts in $L(v')$. Hence, $p=1$, as claimed.
\end{proof}

\end{document}